\documentclass[sort, compress]{elsarticle}
\usepackage[margin=1in]{geometry}
\usepackage{graphicx}
\usepackage[english]{babel}
\usepackage{amsthm}
\usepackage{amsmath}
\usepackage{amssymb}
\usepackage{dsfont}
\usepackage{xcolor}
\usepackage{color}
\usepackage{listings}
\usepackage{hyperref}
\usepackage{enumitem}
\usepackage[normalem]{ulem}
\usepackage{caption, subcaption}

\usepackage{algorithm}
\usepackage{algpseudocode}
\usepackage{tikz}
\usetikzlibrary{calc}

\NewDocumentCommand{\avm}{o o}{%
  \IfNoValueTF{#1}{%
    \widehat{\operatorname{D}}_{n,m}%
  }{%
    \IfNoValueTF{#2}{\widehat{\operatorname{D}}_{#1, #1}}{\widehat{\operatorname{D}}_{#1,#2}}%
  }%
}

\newcommand{\ind}{\mathds 1}
\newcommand{\Cov}{\operatorname{Cov}}

\newtheorem{definition}{Definition}
\newtheorem{remark}{Remark}
\newtheorem{theorem}{Theorem}
\newtheorem{proposition}[theorem]{Proposition}
\newtheorem{lemma}[theorem]{Lemma}
\numberwithin{equation}{section}

\begin{document}

\begin{frontmatter}

    \title{Wasserstein-Based Test for Empirical Measure\\ Convergence of Dependent Sequences}
    \author[GATE,UMD]{Alexander Yordanov\corref{cor1}} \ead{alexander.yordanov@gate-ai.eu}
    \author[GATE]{Peter Hristov} \ead{petar.hristov@gate-ai.eu}
    \affiliation[GATE]{organization={GATE Institute},
                addressline={5 ``James Boucher''  Blvd.},
                city={Sofia},
                postcode={1164},
                country={Bulgaria}}
    \affiliation[UMD]{organization={Department of Mathematics, University of Maryland},
                addressline={4176 Campus Dr.},
                city={College Park},
                postcode={20742},
                state={MD},
                country={USA}}
    \cortext[cor1]{Corresponding author}
    \date{March 2026}

    \begin{abstract}
        We develop Wasserstein-based hypothesis tests for empirical-measure convergence in stationary dependent sequences. For a known candidate invariant measure, $\mu$, we study the statistic $T_n=\sqrt{n}\,W_1(\hat\mu_n,\mu)$ and establish asymptotic level-$\alpha$ validity under the null, together with consistency under fixed alternatives. When the invariant measure is unknown, we derive the asymptotic law of the pairwise statistic $\sqrt{n}\,W_1(\hat\mu_n^{(i)},\hat\mu_n^{(j)})$ for independent trajectories and obtain a corresponding pairwise test, including Bonferroni control for multiple comparisons. To make this estimation feasible when the long-run covariance is unavailable in closed form, we introduce a finite-grid plug-in estimator and show that Gaussian critical values based on the estimated covariance consistently recover the corresponding oracle fixed-grid estimation. Simulation experiments in both linear and nonlinear dynamical settings illustrate the oracle and plug-in regimes, along with the resulting coverage probability and power.
    \end{abstract}
\end{frontmatter}

\section{Introduction}
Many validation problems in stochastic dynamics are naturally long-run questions: does a trajectory reach the correct statistical steady state, and do time-averaged observables stabilize in a way consistent with the target invariant regime? This problem appears both in time-domain diagnostics and in frequency-domain validation, where one compares empirical distributions induced by dependent time series rather than i.i.d.\ samples. Formally, if $(X_t)_{t\ge 1}$ is stationary ergodic with invariant law $\mu$, then the relevant object to study is the empirical measure $\hat\mu_n$ and its convergence to $\mu$.

For this purpose, a Wasserstein metric offers structural advantages over Kolmogorov-Smirnov (KS) type metrics. In one dimension,
\begin{equation}
    W_1(\hat\mu_n,\mu)=\int_{\mathbb R}|\hat{F}_n(t)-F(t)|\,dt,
\end{equation}
where $F(t)$ and $\hat{F}_n(t)$ are the cumulative distribution functions induced by the measures $\mu$ and $\hat{\mu}_n$ respectively. From this definition, we see that the statistic aggregates deviations across the entire state space and quantifies mass displacement in physical units. By contrast, KS uses a supremum norm and depends only on the single largest vertical gap between distribution functions, which can underweight broad but moderate discrepancies. Thus, when convergence quality is tied to geometry of the state space or to cumulative spectral/steady-state behavior, $W_1$ is often the more informative metric. 

This paper builds on the optimal transport and statistical Wasserstein literature \cite{Villani:2008,Panaretos:2019}, while focusing on the dependent-sequence regime that is central to dynamical systems validation. For i.i.d.\ sampling, sharp non-asymptotic control is available in several settings \cite{Fournier:2022}; for stationary dependent sequences, asymptotic results such as Proposition 3.1 in \cite{Dedecker:2017} provide the key Gaussian-process limit theory. Our goal is to turn this asymptotic theory into practical hypothesis tests for empirical-measure convergence, including a data-driven calibration scheme when the long-run covariance is unknown.

Our first contribution is a one-sample test for a known candidate invariant measure $\mu$, based on $T_n=\sqrt n\,W_1(\hat\mu_n,\mu)$, with asymptotic coverage probability and consistency results under fixed alternatives. Our second contribution addresses the practically common case where $\mu$ is unknown: we analyze pairwise distances $\sqrt n\,W_1(\hat\mu_n^{(i)},\hat\mu_n^{(j)})$ across independent trajectories, derive their asymptotic law under the null, and formulate a pairwise testing procedure with Bonferroni correction for multiple comparisons. Our third contribution makes this pairwise procedure implementable when the long-run covariance is not known analytically: we introduce a finite-grid flat-top plug-in estimator of the covariance, prove fixed-grid consistency, and show that critical values built from the estimated covariance consistently recover the corresponding oracle (that is, known theoretical covariance) fixed-grid calibration.

The remainder of the paper is organized as follows. Section 2.1 develops the one-sample Wasserstein test when the invariant measure is known. Section 2.2 presents the pairwise test and its asymptotic guarantees when the invariant measure is unknown. Section 2.3 establishes fixed-grid consistency of the plug-in estimator for the long-run covariance and shows that the resulting Gaussian critical values consistently estimate the corresponding oracle fixed-grid critical values. Section 3 reports simulation studies in linear and nonlinear systems (including double pendulum experiments) illustrating coverage probability and power. Section 4 concludes the exposition, and the appendix provides implementation details for the double pendulum simulation.

\section{Main Results}
\subsection{Testing Empirical Measure Convergence to a Known Invariant Measure}

For the process $(X_t)_{t\geq 1}^n$, denote the $i$-th empirically observed path by $X_t^{(i)},\quad i=1, \dots, M$ and define its empirical measure 
\begin{equation}
    \hat{\mu}_n^{(i)} (\cdot) = \frac{1}{n} \sum_{t=1}^{n} \delta_{X_t^{(i)}} (\cdot ) 
\end{equation}
where $\delta_{X_t^{(i)}}$ is a Dirac delta measure. 

If the system is ergodic with a unique invariant measure $\mu$, then for every $i$,
$\hat{\mu}_n^{(i)}\to\mu$ almost surely as $n\to\infty$. Under finite first-moment assumptions,
\begin{equation}
    W_1 (\hat\mu_n^{(i)}, \mu) \overset{a.s.}{\to} 0
\end{equation}
Moreover, Proposition 3.1 in \cite{Dedecker:2017} yields
\begin{equation}
\sqrt{n} \, W_1 (\hat\mu_n^{(i)}, \mu) \overset{D}{\to} \; \int_{\mathbb{R}} |G (t) | dt =: \mathbb{G},
\end{equation}
where $G$ is a centered Gaussian process with covariance function defined as follows: 
\begin{multline}
    \operatorname{Cov} \left(\int f (t) G(t) dt, \int g(t) G(t) dt \right)\\ =
    \sum_{k\in \mathbb Z} \mathbb E \left( \iint f(t)g(s) (\ind \{X_0\leq t\} - F(t))(\ind \{X_k\leq s\} - F(s))dt \, ds  \right)
\label{eq:cov-G}
\end{multline}
where $f, g \in \mathbb{L}_\infty (\mu)$ are ``test'' functions.

This motivates the following one-sample test against a known invariant measure.

\begin{definition}[Wasserstein test for one invariant measure]\label{def:one-sample-test}
Let $(X_t)_{t=1}^n$ be a stationary ergodic trajectory with empirical measure
$\hat\mu_n=\frac1n\sum_{t=1}^n\delta_{X_t}$, and let $\mu$ be a candidate invariant measure. Consider the hypotheses:
\begin{equation}
H_0: \hat \mu_n \to \mu \qquad vs. \qquad
H_A: \hat\mu_n \to \nu \neq \mu.
\end{equation}
Define
\begin{equation}
T_n := \sqrt{n} W_1 (\hat \mu_n, \mu).
\end{equation}
For $\alpha\in(0,1)$, let $q_{1-\alpha}$ denote the $(1-\alpha)$-quantile of $\mathbb G$.
Reject $H_0$ whenever
\begin{equation}
T_n>q_{1-\alpha}.
\end{equation}
\end{definition}

\begin{proposition}[Asymptotic coverage under $H_0$]\label{prop:one-sample-coverage}
Under the same assumptions as those in Proposition 3.1 in \cite{Dedecker:2017} and assuming that
$\mathbb P(\mathbb G=q_{1-\alpha})=0$. Then the test in Definition~\ref{def:one-sample-test} has asymptotic level $\alpha$, i.e.,
\begin{equation}
\lim_{n\to\infty}\mathbb P_{H_0}(T_n>q_{1-\alpha})=\alpha,
\end{equation}
equivalently,
\begin{equation}
\lim_{n\to\infty}\mathbb P_{H_0}(T_n\le q_{1-\alpha})=1-\alpha.
\end{equation}
\end{proposition}

\begin{proof}
Under $H_0$, Proposition 3.1 in \cite{Dedecker:2017} gives
$T_n \to \mathbb G=\int |G(t)|dt$. By the definition of $q_{1-\alpha}$ and the no-atom condition at that quantile,
\begin{equation}
\lim_{n\to\infty}\mathbb P_{H_0}(T_n>q_{1-\alpha})
=\mathbb P(\mathbb G>q_{1-\alpha})=\alpha.
\end{equation}
The acceptance probability therefore converges to $1-\alpha$.
\end{proof}

\begin{proposition}[Power under fixed alternatives]\label{prop:one-sample-power}
Suppose under $H_A$ the empirical measure satisfies $\hat\mu_n\to\nu$ almost surely, with
$W_1(\nu,\mu)>0$. Then the test in Definition~\ref{def:one-sample-test} is consistent:
\begin{equation}
\mathbb P_{H_A}(T_n>q_{1-\alpha})\to 1.
\end{equation}
\end{proposition}

\begin{proof}
By continuity of $W_1$ with respect to weak convergence (under finite first moments),
\begin{equation}
W_1(\hat\mu_n,\mu)\to W_1(\nu,\mu)=:c>0\qquad\text{a.s.}
\end{equation}
Hence $T_n=\sqrt n\,W_1(\hat\mu_n,\mu)\to\infty$ almost surely as $n\to\infty$. Since $q_{1-\alpha}<\infty$ is fixed,
$\mathbb P_{H_A}(T_n>q_{1-\alpha})\to 1$.
\end{proof}

\subsection{Pairwise Convergence Testing with Two Empirical Measures}

In many applications, the invariant measure is unknown analytically, even though multiple trajectories can be simulated or observed from different initial conditions. In that setting, it is natural to compare empirical measures across trajectory pairs. We therefore study the asymptotic distribution of the pairwise Wasserstein distance and then build a test from that limit.

\begin{theorem}[Asymptotic distribution of the two-sample Wasserstein distance]\label{thm:pairwise-w1-limit}
Let $(X_k^{(i)})_{k\in\mathbb{Z}}$ and $(X_k^{(j)})_{k\in\mathbb{Z}}$ be independent stationary ergodic real-valued sequences with common marginal distribution function $F$, and define
\begin{equation}
\hat\mu_n^{(i)}=\frac1n\sum_{k=1}^n \delta_{X_k^{(i)}},
\qquad
\hat\mu_n^{(j)}=\frac1n\sum_{k=1}^n \delta_{X_k^{(j)}},
\qquad
W_{ij,n}:=W_1(\hat\mu_n^{(i)},\hat\mu_n^{(j)}).
\end{equation}
Assume the joint convergence
\begin{equation}
\left(\sqrt{n}\bigl(F_n^{(i)}-F\bigr),\sqrt{n}\bigl(F_n^{(j)}-F\bigr)\right)
\to
\bigl(G^{(i)},G^{(j)}\bigr)
\quad\text{in }L^1(\mathbb{R})\times L^1(\mathbb{R}),
\end{equation}
where $\bigl(G^{(i)},G^{(j)}\bigr)$ is centered Gaussian, each marginal process has the same covariance structure as in \eqref{eq:cov-G}, and $G^{(i)}$ and $G^{(j)}$ are independent copies.
Then
\begin{equation}
\sqrt{n}\,W_{ij,n}
\to
\int_{\mathbb{R}} \bigl|G^{(i)}(t)-G^{(j)}(t)\bigr|\,dt.
\end{equation}
\end{theorem}

\begin{proof}
For probability measures on $\mathbb{R}$, the $1$-Wasserstein distance admits the representation
\begin{equation}
W_1(\nu_1,\nu_2)=\int_{\mathbb{R}} \bigl|F_{\nu_1}(t)-F_{\nu_2}(t)\bigr|\,dt,
\end{equation}
where $F_{\nu_1}$ and $F_{\nu_2}$ are the distribution functions of $\nu_1$ and $\nu_2$. Applying this identity with $\nu_1=\hat\mu_n^{(i)}$ and $\nu_2=\hat\mu_n^{(j)}$, we obtain
\begin{equation}
W_{ij,n}
=
\int_{\mathbb{R}} \bigl|F_n^{(i)}(t)-F_n^{(j)}(t)\bigr|\,dt.
\end{equation}
Multiplying by $\sqrt{n}$ yields
\begin{equation}
\sqrt{n}\,W_{ij,n}
=
\int_{\mathbb{R}}
\left|
\sqrt{n}\bigl(F_n^{(i)}(t)-F(t)\bigr)
-
\sqrt{n}\bigl(F_n^{(j)}(t)-F(t)\bigr)
\right|dt.
\end{equation}
In terms of the processes
\begin{equation}
G_n^{(i)}=\sqrt{n}(F_n^{(i)}-F),
\qquad
G_n^{(j)}=\sqrt{n}(F_n^{(j)}-F),
\end{equation}
this becomes
\begin{equation}
\sqrt{n}\,W_{ij,n}
=
\|G_n^{(i)}-G_n^{(j)}\|_{L^1}.
\end{equation}
Now define the continuous mapping
\begin{equation}
\Phi:L^1(\mathbb{R})\times L^1(\mathbb{R})\to\mathbb{R},
\qquad
\Phi(f,g):=\|f-g\|_{L^1}.
\end{equation}
Since by assumption
\begin{equation}
\bigl(G_n^{(i)},G_n^{(j)}\bigr)
\to
\bigl(G^{(i)},G^{(j)}\bigr)
\quad\text{in }L^1(\mathbb{R})\times L^1(\mathbb{R}),
\end{equation}
it follows from the continuous mapping theorem that,
\begin{equation}
\sqrt{n}\,W_{ij,n}
=
\Phi\bigl(G_n^{(i)},G_n^{(j)}\bigr)
\to
\Phi\bigl(G^{(i)},G^{(j)}\bigr)
=
\|G^{(i)}-G^{(j)}\|_{L^1}.
\end{equation}
Equivalently,
\begin{equation}
\sqrt{n}\,W_{ij,n}
\to
\int_{\mathbb{R}} \bigl|G^{(i)}(t)-G^{(j)}(t)\bigr|\,dt,
\end{equation}
which completes the proof.
\end{proof}

\begin{definition}[Asymptotic pairwise Wasserstein test]\label{def:pairwise-test}
For one observed pair of trajectories, define
\begin{equation}
T_{ij,n}:=\sqrt{n}\,W_1\bigl(\hat\mu_n^{(i)},\hat\mu_n^{(j)}\bigr),
\qquad
H_0:\mu^{(i)}=\mu^{(j)}
\quad\text{vs.}\quad
H_A:\mu^{(i)}\neq\mu^{(j)}.
\end{equation}
Let
\begin{equation}
Z_{ij}:=\int_{\mathbb R}\bigl|G^{(i)}(t)-G^{(j)}(t)\bigr|\,dt,
\end{equation}
and let $q_{1-\alpha}^{(ij)}$ be the $(1-\alpha)$-quantile of $Z_{ij}$. Reject $H_0$ whenever
\begin{equation}
T_{ij,n}>q_{1-\alpha}^{(ij)}.
\end{equation}
\end{definition}

\begin{proposition}[Asymptotic coverage under $H_0$]\label{prop:pairwise-coverage}
Assume the conditions of Theorem~\ref{thm:pairwise-w1-limit} and
$\mathbb P\bigl(Z_{ij}=q_{1-\alpha}^{(ij)}\bigr)=0$. Then the test in Definition~\ref{def:pairwise-test} has asymptotic level $\alpha$:
\begin{equation}
\lim_{n\to\infty}\mathbb P_{H_0}\!\left(T_{ij,n}>q_{1-\alpha}^{(ij)}\right)=\alpha.
\end{equation}
\end{proposition}

\begin{proof}
By Theorem~\ref{thm:pairwise-w1-limit}, under $H_0$ we have
$T_{ij,n}\to Z_{ij}$. Therefore,
\begin{equation}
\lim_{n\to\infty}\mathbb P_{H_0}\!\left(T_{ij,n}>q_{1-\alpha}^{(ij)}\right)
=\mathbb P\!\left(Z_{ij}>q_{1-\alpha}^{(ij)}\right)=\alpha,
\end{equation}
using the no-atom condition at the quantile.
\end{proof}

\begin{proposition}[Power under fixed alternatives]\label{prop:pairwise-power}
Suppose under $H_A$,
\begin{equation}
W_1\bigl(\mu^{(i)},\mu^{(j)}\bigr)>0,
\qquad
\hat\mu_n^{(i)}\to\mu^{(i)}\ \text{a.s.},
\qquad
\hat\mu_n^{(j)}\to\mu^{(j)}\ \text{a.s.}
\end{equation}
Then the test in Definition~\ref{def:pairwise-test} is consistent:
\begin{equation}
\mathbb P_{H_A}\!\left(T_{ij,n}>q_{1-\alpha}^{(ij)}\right)\to 1.
\end{equation}
\end{proposition}

\begin{proof}
By continuity of $W_1$ (under finite first moments),
\begin{equation}
W_1\bigl(\hat\mu_n^{(i)},\hat\mu_n^{(j)}\bigr)
\to
W_1\bigl(\mu^{(i)},\mu^{(j)}\bigr)=:c>0
\quad\text{a.s.}
\end{equation}
Hence $T_{ij,n}=\sqrt n\,W_1(\hat\mu_n^{(i)},\hat\mu_n^{(j)})\to\infty$ almost surely, so for fixed $q_{1-\alpha}^{(ij)}$,
$\mathbb P_{H_A}(T_{ij,n}>q_{1-\alpha}^{(ij)})\to 1$.
\end{proof}

\begin{remark}[Multiple pairwise tests and Bonferroni correction]
If one performs pairwise tests over a collection of $K$ trajectory pairs, using level $\alpha$ for each individual test inflates the overall type-I error. A simple way to maintain overall asymptotic coverage (family-wise error rate) at level $\alpha$ is to use a Bonferroni correction: test each pair at level
\begin{equation}
\alpha_{\mathrm{pair}}=\frac{\alpha}{K},
\end{equation}
that is, reject pair $(i,j)$ only when
\begin{equation}
T_{ij,n}>q_{1-\alpha/K}^{(ij)}.
\end{equation}
Then, asymptotically, the probability of at least one false rejection across all $K$ pairs is at most $\alpha$.
\end{remark}

\subsection{Plug-In Covariance Estimator on a Fixed Grid}

In many applications, the covariance structure is not available in closed form and must therefore be estimated from the data. In this subsection, we work on a fixed grid $u_1,\dots,u_m$ and separate two issues: the covariance-estimation error on that grid, and the additional grid-discretization error incurred when approximating the continuous Wasserstein limit by a finite sum. We prove only the first step here and show that the resulting critical values consistently recover the corresponding oracle fixed-grid critical values. Developing this fixed-grid proof into a proof of the exact continuous Wasserstein limit requires a separate refinement argument with $m=m_n\to\infty$, which we do not address in this paper.

\begin{definition}[Finite-grid plug-in estimator of the long-run covariance matrix]\label{def:plug-in-cov-est}
Let $u_1<\dots<u_m\in\mathbb{R}$ be a finite grid of fixed size \(m\), and let
\begin{equation}
\hat F_n\!\left(u_i\right) = \frac{1}{n}\sum_{t=1}^n \ind\{X_t\le u_i\},
\qquad i=1,\dots,m.
\end{equation}
For all $t\leq n$, define the centered indicator vector:
\begin{equation}
Y_t^{(m)}
:=
\Bigl(
\ind\{X_t\le u_1\}-\hat F_n(u_1),
\dots,
\ind\{X_t\le u_m\}-\hat F_n(u_m)
\Bigr)^\top
\in\mathbb{R}^{m}.
\end{equation}
For lags \(k=0,1,\dots,L_n\) where $L_n$ is the maximum lag cutoff (also referred to as the bandwidth in the literature), define the empirical lag-\(k\) cross-covariance matrices of the vector process $Y_t^{(m)}$ by:
\begin{equation}
\hat\Gamma_{n,k}^{(m)}
:=
\frac{1}{n-k}\sum_{t=1}^{n-k}
Y_t^{(m)}\bigl(Y_{t+k}^{(m)}\bigr)^\top .
\end{equation}
Let \(w_{n,k}\in\mathbb{R}\) be a sequence of lag-window weights with \(w_{n,0}=1\).
The finite-grid plug-in estimator of the long-run covariance matrix follows the structure of a Newey-West estimator (\cite{newey:1986}) and is defined as follows: 
\begin{equation}
\widehat\Sigma_{n,m}
:=
\hat\Gamma_{n,0}^{(m)}
+
\sum_{k=1}^{L_n}
w_{n,k}\Bigl(
\hat\Gamma_{n,k}^{(m)}+\bigl(\hat\Gamma_{n,k}^{(m)}\bigr)^\top
\Bigr).
\end{equation}

\end{definition}

Now, consider the two vectors evaluated on the grid above -- one using the theoretical distribution and the other the empirical distribution:
\begin{equation}
\theta
:=
\bigl(F(u_1),\dots,F(u_m)\bigr)^\top,
\qquad
\hat\theta_n
:=
\bigl(\hat F_n(u_1),\dots,\hat F_n(u_m)\bigr)^\top.
\end{equation}
For a general \(\vartheta=(\vartheta_1,\dots,\vartheta_m)^\top\in\mathbb R^m\), define the following centered indicator process:
\begin{equation}
V_t(\vartheta)
:=
\bigl(
\ind\{X_t\le u_1\}-\vartheta_1,\dots,
\ind\{X_t\le u_m\}-\vartheta_m
\bigr)^\top.
\end{equation}
Evaluating at $\vartheta = \hat{\theta}_n$, we obtain $Y_t^{(m)}=V_t(\hat\theta_n)$ from above. Additionally, 
\begin{equation}
V_t(\theta)
=
\bigl(
\ind\{X_t\le u_1\}-F(u_1),\dots,
\ind\{X_t\le u_m\}-F(u_m)
\bigr)^\top.
\end{equation}
The theoretical covariance at lag $k$ can then be defined as:
\begin{equation}
\Gamma_k
:=
\mathbb E\bigl[V_0(\theta)V_k(\theta)^\top\bigr],
\qquad k\in\mathbb Z.
\end{equation}
The associated finite-grid, long-run covariance matrix is
\begin{equation}
\Sigma_m(a,b)
:=
\sum_{k\in\mathbb Z}
\Cov\Bigl(
\ind\{X_0\le u_a\}-F(u_a),
\ind\{X_k\le u_b\}-F(u_b)
\Bigr),
\qquad 1\le a,b\le m.
\end{equation}
Equivalently,
\begin{equation}
\Sigma_m
=
\sum_{k\in\mathbb Z}\Gamma_k.
\end{equation}

\begin{proposition}[Consistency of the finite-grid flat-top plug-in estimator]\label{prop:plugin-covariance}
Assume the following.

\begin{enumerate}
\item[(i)] \(\{V_t(\theta)\}_{t\in\mathbb Z}\) is second-order stationary and satisfies the
standard weak-dependence and moment assumptions ensuring consistency of kernel HAC
estimators; for example, Assumptions A, B, and C in \cite{andrews:1988}.

\item[(ii)] The finite-grid long-run covariance matrix $\Sigma_m$ is well defined; for example,
\begin{equation}
\sum_{k\in\mathbb Z}\|\Gamma_k\|_{\max}<\infty.
\end{equation}

\item[(iii)] The lag-window weights are of flat-top form
\begin{equation}
w_{n,k}=\lambda(k/L_n),\qquad k=0,1,\dots,L_n,
\end{equation}
where \(\lambda:\mathbb R\to[-1,1]\) is even, satisfies \(\lambda(0)=1\), is continuous at
\(0\) and all but finitely many points, is square-integrable, and is flat near the origin, i.e.
\begin{equation}
\lambda(x)=1\qquad\text{for } |x|\le c
\end{equation}
for some \(c>0\).

\item[(iv)] The bandwidth satisfies
\begin{equation}
L_n\to\infty,
\qquad
\frac{L_n}{n}\to 0.
\end{equation}
\end{enumerate}
Then, for each fixed \(m\),
\begin{equation}
\|\widehat\Sigma_{n,m}-\Sigma_m\|_{\max}\xrightarrow[n\to\infty]{\mathbb P}0.
\end{equation}
\end{proposition}

\begin{proof}
The proof proceeds in two steps. First, we compare \(\widehat\Sigma_{n,m}\) with the corresponding
oracle estimator centered at the true vector \(\theta\). Second, we show that this
oracle estimator converges in probability to the target covariance matrix \(\Sigma_m\). For
\(k=0,1,\dots,L_n\), define the true-centered lag-covariance estimator by
\begin{equation}
\tilde\Gamma_{n,k} := \frac{1}{n-k}\sum_{t=1}^{n-k}V_t(\theta)V_{t+k}(\theta)^\top,
\end{equation}
and
\begin{equation} 
\label{eq:sigma_tilde_def}
\widetilde\Sigma_{n,m}
:=
\tilde\Gamma_{n,0}
+
\sum_{k=1}^{L_n}
w_{n,k}\bigl(\tilde\Gamma_{n,k}+\tilde\Gamma_{n,k}^\top\bigr).
\end{equation}
Write \(\Delta_n:=\hat\theta_n-\theta\). Since \(V_t(\vartheta)\) is affine in \(\vartheta\),
its first-order Taylor expansion around \(\theta\) is exact:
\begin{equation}
V_t(\hat\theta_n)
=
V_t(\theta)
+
\frac{\partial}{\partial\vartheta}V_t(\theta)\,\Delta_n
=
V_t(\theta)-\Delta_n,
\qquad
\frac{\partial}{\partial\vartheta}V_t(\vartheta)=-\mathbb{I}_m.
\end{equation}
For each lag \(k=0,1,\dots,L_n\), define
\begin{equation}
\Gamma_{n,k}(\vartheta)
:=
\frac{1}{n-k}\sum_{t=1}^{n-k}V_t(\vartheta)V_{t+k}(\vartheta)^\top.
\end{equation}
Then \(\Gamma_{n,k}(\hat\theta_n)=\hat\Gamma_{n,k}^{(m)}\) and
\(\Gamma_{n,k}(\theta)=\tilde\Gamma_{n,k}\), and expanding at \(\theta\) gives
\begin{equation}
\Gamma_{n,k}(\hat\theta_n)-\Gamma_{n,k}(\theta)
=
-\frac{1}{n-k}\sum_{t=1}^{n-k}V_t(\theta)\Delta_n^\top
-\frac{1}{n-k}\sum_{t=1}^{n-k}\Delta_n V_{t+k}(\theta)^\top
+\Delta_n\Delta_n^\top.
\end{equation}
Moreover, the linear terms are centered, since
\begin{equation}
\mathbb E\!\left[V_t(\theta)\frac{\partial}{\partial\vartheta}V_{t-j}(\theta)^\top\right]
=
-\mathbb E[V_t(\theta)]
=0.
\qquad\text{for all }j,
\end{equation}
Now observe that \(\Delta_n=O_{\mathbb P}(n^{-1/2})\) by the \(\sqrt n\)-consistency of
\(\hat\theta_n\). For one of the linear terms,
\begin{equation}
\left\|\frac{1}{n-k}\sum_{t=1}^{n-k}V_t(\theta)\Delta_n^\top\right\|_{\max}
\le
\|\Delta_n\|_{\max}
\left\|\frac{1}{n-k}\sum_{t=1}^{n-k}V_t(\theta)\right\|_{\max} .
\end{equation}
Under the same weak-dependence assumptions, the centered average
\(\frac{1}{n-k}\sum_{t=1}^{n-k}V_t(\theta)\) is \(O_{\mathbb P}((n-k)^{-1/2})\); since
\(k\le L_n\) and \(L_n/n\to0\), we have \(n-k\asymp n\) uniformly over the retained lags.
Therefore
\(\frac{1}{n-k}\sum_{t=1}^{n-k}V_t(\theta)\Delta_n^\top=O_{\mathbb P}(n^{-1})\)
for each lag, and after summing over \(k\le L_n\) with bounded weights its total contribution is
\(O_{\mathbb P}(L_n/n)=o_{\mathbb P}(1)\). The same bound applies to the second linear term,
while \(\Delta_n\Delta_n^\top=O_{\mathbb P}(n^{-1})\). Hence the difference between the
plug-in and oracle lag-covariance estimators is a sum of asymptotically negligible terms.
Summing over \(k\le L_n\) with bounded weights yields the same conclusion for the HAC
estimators; cf. the standard plug-in HAC arguments in \cite{andrews:1988,politis:2011}.
In particular,
\begin{equation}
    \|\widehat\Sigma_{n,m}-\widetilde\Sigma_{n,m}\|_{\max}\xrightarrow{\mathbb P}0.
    \label{eq:finite-grid-conv}
\end{equation}

It remains to show that \(\widetilde\Sigma_{n,m}\to\Sigma_m\) in probability.
Fix \(1\le a,b\le m\). Write
\begin{equation}
V_t^{(a)}(\theta):=\ind\{X_t\le u_a\}-F(u_a),
\qquad
V_t^{(b)}(\theta):=\ind\{X_t\le u_b\}-F(u_b).
\end{equation}
Then $\widetilde\Sigma_{n, m}$ from (\ref{eq:sigma_tilde_def}) is precisely the kernel HAC estimator of the long-run covariance $\Sigma_m(a,b)$.

Since \(\lambda\) is an admissible lag window and the bandwidth satisfies \(L_n\to\infty\)
and \(L_n/n\to0\), standard consistency results for kernel HAC estimators imply that
\begin{equation}
\widetilde\Sigma_{n,m}(a,b)\xrightarrow{\mathbb P}\Sigma_m(a,b)
\qquad\text{for each fixed }a,b;
\end{equation}
see \cite{andrews:1988}. Therefore, because \(m\) is fixed,
\begin{equation}
    \|\widetilde\Sigma_{n,m}-\Sigma_m\|_{\max}\xrightarrow{\mathbb P}0.
    \label{eq:plug-in-conv}
\end{equation}
Finally, by the triangle inequality,
\begin{equation}
\|\widehat\Sigma_{n,m}-\Sigma_m\|_{\max}
\le
\|\widehat\Sigma_{n,m}-\widetilde\Sigma_{n,m}\|_{\max}
+
\|\widetilde\Sigma_{n,m}-\Sigma_m\|_{\max}.
\end{equation}
Combining this with (\ref{eq:finite-grid-conv}) and (\ref{eq:plug-in-conv}) yields
\begin{equation}
\|\widehat\Sigma_{n,m}-\Sigma_m\|_{\max}\xrightarrow{\mathbb P}0,
\end{equation}
which proves the claim.
\end{proof}

\begin{remark}
The role of the flat-top window is not to obtain mere consistency, which is attainable for general kernel HAC \cite{andrews:1988, newey:1986}, but to improve the
bias and, when additional smoothness is present, the rate of convergence; see
\cite{politis:1999,politis:2011}.
\end{remark}

\begin{lemma}[Plug-in Gaussian critical values converge to the oracle fixed-grid critical values]\label{lem:plugin-quantile}
Let $G^{(i)}$ and $G^{(j)}$ be the Gaussian processes from Theorem~\ref{thm:pairwise-w1-limit}. Fix the grid size $m$, let $\Delta_\ell:=u_{\ell+1}-u_\ell$ for $\ell=1,\dots,m-1$, and define the grid approximation
\begin{equation}
Z_m^{(ij)}
:=
\sum_{\ell=1}^{m-1}
\Delta_\ell\,
\bigl|G^{(i)}(u_\ell)-G^{(j)}(u_\ell)\bigr|.
\end{equation}
Let $q_{m,1-\alpha}^{(ij)}$ be the $(1-\alpha)$-quantile of $Z_m^{(ij)}$. Assume the hypotheses of Proposition~\ref{prop:plugin-covariance}, and suppose
\begin{equation}
\mathbb P\bigl(Z_m^{(ij)}=q_{m,1-\alpha}^{(ij)}\bigr)=0.
\end{equation}
Let
\begin{equation}
\widehat{\mathbf G}_{n,m}^{(r)}
:=
\bigl(\widehat G_{n,m}^{(r)}(u_1),\dots,\widehat G_{n,m}^{(r)}(u_m)\bigr)^\top,
\qquad r\in\{i,j\},
\end{equation}
be, conditionally on the data, independent centered Gaussian vectors in $\mathbb R^m$ with covariance matrix $\widehat\Sigma_{n,m}$, and define
\begin{equation}
\widehat Z_{n,m}^{(ij)}
:=
\sum_{\ell=1}^{m-1}
\Delta_\ell\,
\bigl|\widehat G_{n,m}^{(i)}(u_\ell)-\widehat G_{n,m}^{(j)}(u_\ell)\bigr|.
\end{equation}
Let $\widehat q_{n,m,1-\alpha}^{(ij)}$ denote the conditional $(1-\alpha)$-quantile of $\widehat Z_{n,m}^{(ij)}$. Then, for fixed $m$,
\begin{equation}
\widehat q_{n,m,1-\alpha}^{(ij)}\xrightarrow{\mathbb P}q_{m,1-\alpha}^{(ij)}.
\end{equation}
\end{lemma}

\begin{proof}
By Proposition~\ref{prop:plugin-covariance},
\begin{equation}
\widehat\Sigma_{n,m}\xrightarrow{\mathbb P}\Sigma_m.
\end{equation}
Since $m$ is fixed, the conditional law of the Gaussian difference vector $\widehat{\mathbf G}_{n,m}^{(i)}-\widehat{\mathbf G}_{n,m}^{(j)}$, whose covariance is $2\widehat\Sigma_{n,m}$, converges weakly in probability to the law of
\begin{equation}
\mathbf G_m^{(i)}-\mathbf G_m^{(j)},
\qquad
\mathbf G_m^{(r)}:=\bigl(G^{(r)}(u_1),\dots,G^{(r)}(u_m)\bigr)^\top,
\quad r\in\{i,j\},
\end{equation}
whose covariance is $2\Sigma_m$. Applying the continuous mapping
\begin{equation}
x\mapsto \sum_{\ell=1}^{m-1}\Delta_\ell |x_\ell|
\end{equation}
from $\mathbb R^m$ to $\mathbb R$ therefore yields convergence of the conditional law of $\widehat Z_{n,m}^{(ij)}$ to that of $Z_m^{(ij)}$. Since $\mathbb P\bigl(Z_m^{(ij)}=q_{m,1-\alpha}^{(ij)}\bigr)=0$, standard quantile continuity implies
\begin{equation}
\widehat q_{n,m,1-\alpha}^{(ij)}\xrightarrow{\mathbb P}q_{m,1-\alpha}^{(ij)}.
\end{equation}
This proves the claim.
\end{proof}

\begin{remark}[Fixed-grid estimation versus the continuous Wasserstein limit]
Proposition~\ref{prop:plugin-covariance} and Lemma~\ref{lem:plugin-quantile} address only the fixed-grid covariance-estimation problem: for each fixed $m$, the random matrix $\widehat\Sigma_{n,m}$ consistently estimates $\Sigma_m$, and the plug-in Gaussian quantile consistently estimates the corresponding oracle quantile $q_{m,1-\alpha}^{(ij)}$ for the discretized functional $Z_m^{(ij)}$. In particular, these fixed-grid results do not by themselves imply
\begin{equation}
\widehat q_{n,m,1-\alpha}^{(ij)}\xrightarrow{\mathbb P}q_{1-\alpha}^{(ij)}
\end{equation}
for the continuous Wasserstein limit
\begin{equation}
Z_{ij}=\int_{\mathbb R}\bigl|G^{(i)}(t)-G^{(j)}(t)\bigr|\,dt.
\end{equation}
To obtain asymptotic validity for the exact continuous $W_1$ statistic, one would need an additional grid-refinement scheme $m=m_n\to\infty$ together with assumptions ensuring that the discretization error vanishes, for example $Z_{m_n}^{(ij)}\Rightarrow Z_{ij}$. We leave this refinement to future work.
\end{remark}

\section{Simulation Results}
In this section, we report three simulation studies that illustrate both: (i) coverage probability under the null hypothesis, where trajectories converge to the same invariant distribution, and (ii) power under fixed alternatives, where trajectories converge to different invariant distributions. In all three, we focus only on the pairwise distribution of the Wasserstein statistic. We refer to the ``convergent case'' when the scaled empirical pairwise statistics converge to the theoretical Gaussian limit, and to the ``divergent case'' when they do not. 

\subsection{Known Finite Covariance}
We begin with a stationary $\operatorname{MA}(3)$ process with coefficients $\theta=(\theta_1,\theta_2,\theta_3)=(0.6,0.4,0.2)$:
\begin{equation*}
    X_t = \mu + \varepsilon_t + \theta_1\varepsilon_{t-1}+\theta_2\varepsilon_{t-2}+\theta_3\varepsilon_{t-3},
    \qquad \varepsilon_t \stackrel{iid}{\sim} \mathcal{N}(0,1).
\end{equation*}

For the convergent setting (Figure~\ref{fig:ma_comparison}(a)), all 100 trajectories are generated with mean $\mu=0$. For the divergent setting (Figure~\ref{fig:ma_comparison}(b)), we split trajectories into two groups with means $\mu^{(0)}=0$ and $\mu^{(1)}=0.5$, and compute pairwise Wasserstein distances only within groups. In both settings, the finite-lag covariance structure is known explicitly:
\begin{equation*}
\gamma(h):=\operatorname{Cov}(X_t,X_{t+h})=
\begin{cases}
\sigma_\varepsilon^2\displaystyle\sum_{j=0}^{3-|h|}\psi_j\psi_{j+|h|}, & |h|\le 3,\\
0, & |h|>3,
\end{cases}
\qquad
(\psi_0,\psi_1,\psi_2,\psi_3)=(1,\theta_1,\theta_2,\theta_3),
\end{equation*}
and with $\sigma_\varepsilon^2=1$ and $\theta=(0.6,0.4,0.2)$ this gives
\begin{equation*}
\gamma(0)=1.56,\qquad \gamma(1)=0.92,\qquad \gamma(2)=0.52,\qquad \gamma(3)=0.20,\qquad \gamma(h)=0\ (|h|\ge 4).
\end{equation*}
These closed-form values are used to simulate the Gaussian limit for the pairwise test statistic.
\begin{figure}[h!]
    \centering
    \begin{subfigure}[t]{0.49\textwidth}
        \centering
        \includegraphics[width=\textwidth]{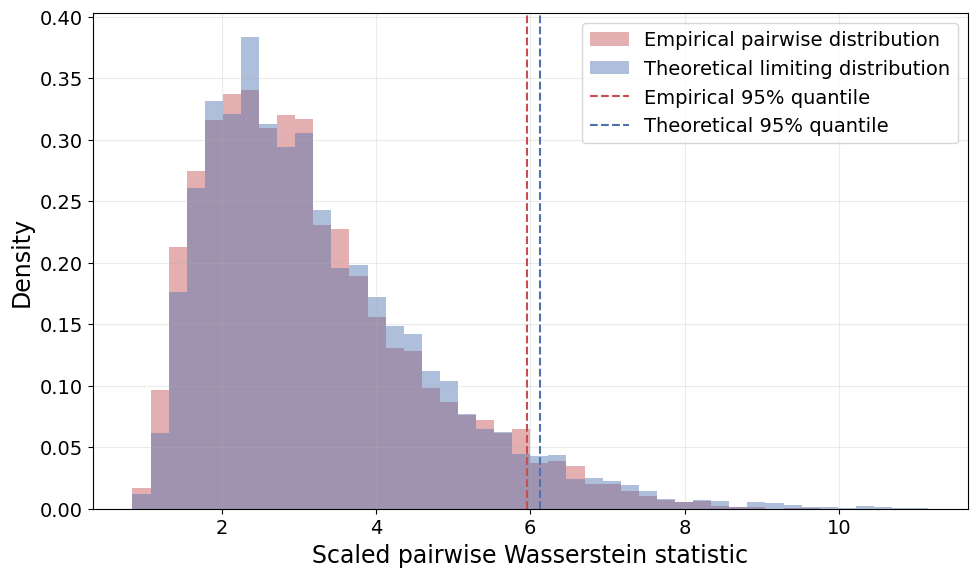}
        \caption{}
    \end{subfigure}
    \hfill
    \begin{subfigure}[t]{0.49\textwidth}
        \centering
        \includegraphics[width=\textwidth]{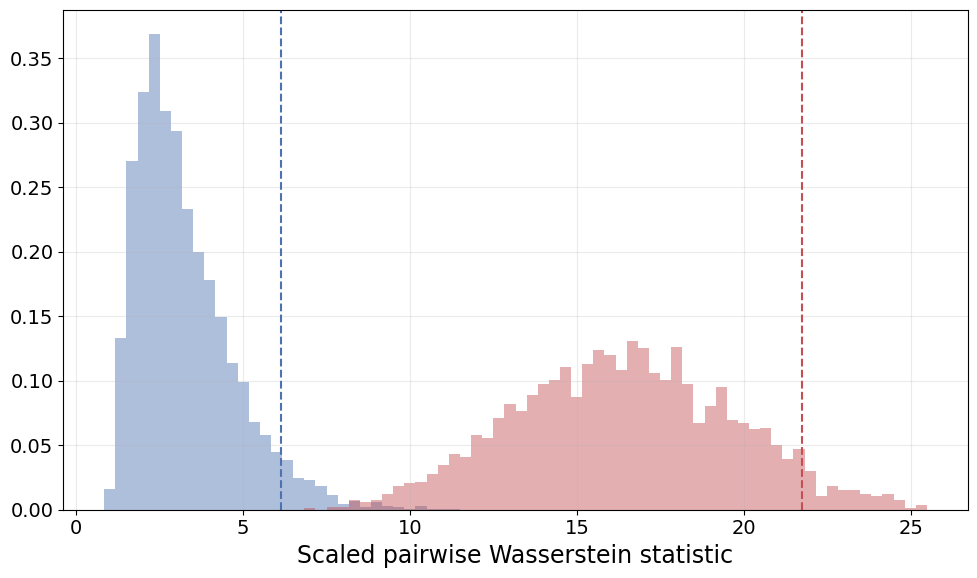}
        \caption{}
    \end{subfigure}
    \caption{Comparison of the distribution of the scaled empirical Wasserstein statistics $\sqrt{n}\,W_1\bigl(\hat\mu_n^{(i)},\hat\mu_n^{(j)}\bigr)$ and the theoretical limiting distribution for an $\operatorname{MA}(3)$ sequence. (a) Null case: trajectories share one invariant distribution; (b) Alternative case: trajectories come from two different invariant distributions.}
    \label{fig:ma_comparison}
\end{figure}

Figure~\ref{fig:ma_comparison}(a) shows close agreement between the empirical distribution of the statistic and its asymptotic approximation under $H_0$, including good alignment of the 95\% quantile. Figure~\ref{fig:ma_comparison}(b) shows clear separation under $H_A$; even for moderate sample sizes (e.g., $n=1000$), the test already exhibits high power. This pattern is reinforced by Figure~\ref{fig:ma_divergent_densities} and Table~\ref{tab:ma_power}: as $n$ increases, the distribution of $\sqrt{n}\,W_1$ shifts to the right, yielding rapidly increasing rejection probabilities and confirming consistency.

\begin{figure}[h!]
    \centering
    \includegraphics[width=0.7\textwidth]{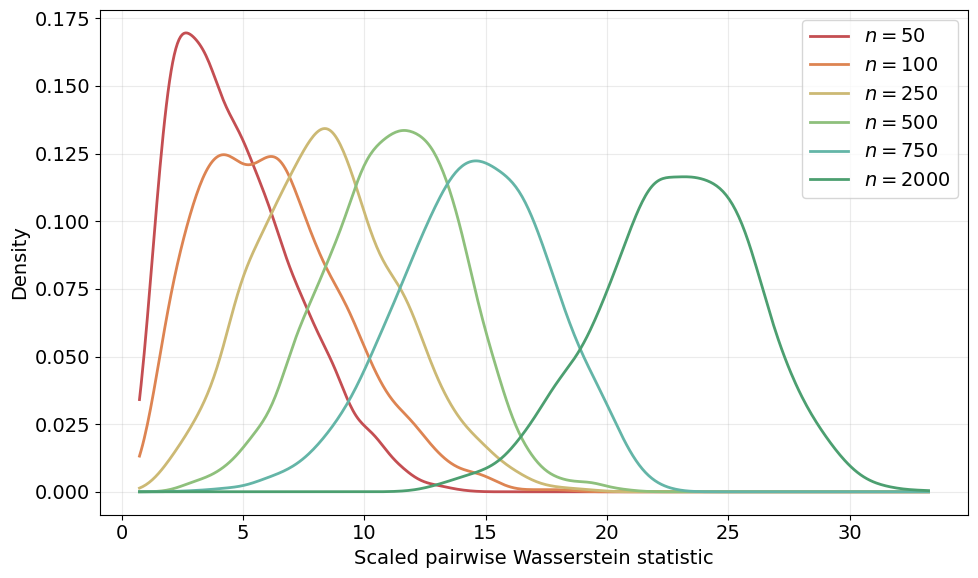}
    \caption{Sampling distributions of $\sqrt{n}\,W_1\!\left(\hat\mu_n^{(i)},\hat\mu_n^{(j)}\right)$ under the alternative in an $\operatorname{MA}(3)$ setting, showing rightward drift as $n$ increases.}
    \label{fig:ma_divergent_densities}
\end{figure}

\begin{table}[h!]
\centering
\begin{equation}
\text{\bf True Rejection Rates}
\end{equation}
\begin{tabular}{cccc}
\hline
\textbf{$n$} & \boldmath$\alpha=0.01$ & \boldmath$\alpha=0.05$ & \boldmath$\alpha=0.10$ \\
\hline
50   & 9.28\%  & 22.24\% & 34.88\% \\
100  & 22.52\% & 41.92\% & 56.76\% \\
250  & 49.64\% & 73.08\% & 84.00\% \\
500  & 82.52\% & 94.52\% & 97.60\% \\
750  & 96.88\% & 99.04\% & 99.72\% \\
1000 & 98.12\% & 99.44\% & 99.80\% \\
\hline
\end{tabular}
\caption{Empirical power (true rejection rate) in the divergent $\operatorname{MA}(3)$ setting, by sample size $n$ and test level $\alpha$.}
\label{tab:ma_power}
\end{table}

\subsection{Known Infinite Covariance}

\begin{figure}[h!]
    \centering
    \includegraphics[width=0.7\textwidth]{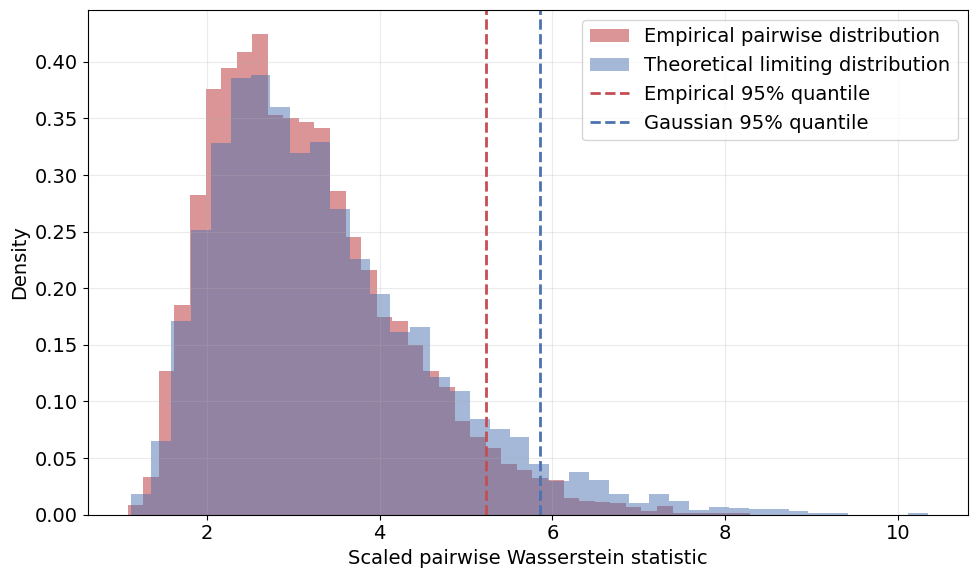}
    \caption{Comparison of the distribution of the scaled empirical Wasserstein statistics $\sqrt{n}\,W_1\bigl(\hat\mu_n^{(i)},\hat\mu_n^{(j)}\bigr)$ and the theoretical limiting distribution for an $\operatorname{ARMA}(5,3)$ sequence in the convergent case where the Wasserstein statistic agrees well with the asymptotic limit distribution under $H_0$.}
    \label{fig:arma_convergent}
\end{figure}

We now consider a process with an infinite, but explicitly known, covariance structure: an $\operatorname{ARMA}(5,3)$ model. It is defined by
\begin{equation*}
\Phi(B)X_t=\Theta(B)\varepsilon_t,
\qquad
\varepsilon_t\stackrel{iid}{\sim}\mathcal N(0,1),
\end{equation*}
where the lag polynomials are
\begin{equation*}
\Phi(B)=1-0.7B+0.25B^2+0.18B^3-0.12B^4+0.08B^5,
\qquad
\Theta(B)=1+0.5B-0.4B^2+0.25B^3.
\end{equation*}

While this backshift representation is standard, it is more useful for our purposes to work with the induced covariance structure. Under stationarity, the process admits a causal MA($\infty$) representation
\begin{equation}
X_t=\sum_{j=0}^{\infty}\psi_j\varepsilon_{t-j},
\qquad
\sum_{j=0}^{\infty}\psi_j z^j=\frac{\Theta(z)}{\Phi(z)},
\end{equation}
which yields the autocovariance function
\begin{equation*}
\gamma(h)=\operatorname{Cov}(X_t,X_{t+h})
=\sum_{j=0}^{\infty}\psi_j\psi_{j+|h|}, 
\qquad h\in\mathbb Z.
\end{equation*}
In contrast to finite-order moving average models, $\gamma(h)$ is nonzero for infinitely many lags, although it decays geometrically.

In the numerical implementation, we do not compute the coefficients $\psi_j$ explicitly. Instead, we exploit the fact that the full autocovariance sequence $\{\gamma(h)\}$ is available in closed form for ARMA models. Using the \texttt{statsmodels} implementation, we obtain
\begin{equation}
\{\gamma(0),\ldots,\gamma(K)\}
\end{equation}
directly via the model-implied autocovariance function, and truncate at a sufficiently large lag $K$ ($200$ in this case).

These autocovariances are then used to construct the covariance kernel of the limiting Gaussian process,
\begin{equation}
\Gamma_K(x,y)
=
\operatorname{Cov}_0(x,y)
+
2\sum_{k=1}^{K}\operatorname{Cov}_k(x,y),
\end{equation}
where each lag-$k$ contribution depends only on the correlation
\begin{equation}
\rho_k=\frac{\gamma(k)}{\gamma(0)}.
\end{equation}
The resulting kernel is discretized on a grid and assembled into a covariance matrix $\Sigma$, from which Gaussian process realizations are simulated.

\medskip

This approach avoids estimating the covariance structure from data and instead leverages the exact model-implied dependence, isolating the effect of finite-sample approximation from covariance misspecification.

Figure~\ref{fig:arma_convergent} exhibits the same qualitative pattern as in the finite-covariance experiment: under the null hypothesis, the empirical distribution of the test statistic tracks its asymptotic approximation closely.

\subsection{Estimated Infinite Covariance (Double Pendulum)}
\begin{figure}[h!]
    \centering
    \begin{tikzpicture}[line cap=round,line join=round,scale=0.5,every node/.style={font=\large}]
        \def\ang{32}
        \def\Lone{4.2}
        \def\Ltwo{3.8}

        \coordinate (O) at (0,0);
        \coordinate (M1) at ({\Lone*cos(-90+\ang)},{\Lone*sin(-90+\ang)});
        \coordinate (M2) at ($(M1)+({\Ltwo*cos(-90+\ang)},{\Ltwo*sin(-90+\ang)})$);

        \draw[line width=2.2pt] (-1.8,0.35) -- (1.8,0.35);
        \foreach \x in {-1.55,-1.15,-0.75,-0.35,0.05,0.45,0.85,1.25} {
            \draw[line width=1.6pt] (\x,0.35) -- ++(0.22,0.28);
        }

        \draw[line width=2pt] (O) -- (M1);
        \draw[line width=2pt] (M1) -- (M2);

        \draw[dashed,line width=1.5pt] (O) -- (0,-5.0);
        \draw[dashed,line width=1.5pt] (M1) -- ($(M1)+(0,-4.2)$);

        \draw[line width=1.6pt,fill=white] (O) circle (0.34);
        \draw[line width=1.6pt,fill=white] (M1) circle (0.34);
        \draw[line width=1.6pt,fill=white] (M2) circle (0.34);

        \draw[line width=1.4pt,->] (0,-2.0) arc[start angle=-90,end angle={-90+\ang},radius=2.0];
        \draw[line width=1.4pt,->] ($(M1)+(0,-1.65)$) arc[start angle=-90,end angle={-90+\ang},radius=1.65];

        \node at ($(O)!0.53!(M1)+(0.8,0.4)$) {$l_1$};
        \node at ($(M1)!0.55!(M2)+(0.85,0.18)$) {$l_2$};
        \node at ($(M1)+(1.1,0.2)$) {$m_1$};
        \node at ($(M2)+(1.1,0.0)$) {$m_2$};
        \node at (0.95,-2.75) {$\theta_1$};
        \node at ($(M1)+(0.7,-2.5)$) {$\theta_2$};
    \end{tikzpicture}
    \caption{Double pendulum geometry used in the long-run covariance estimation experiment.}
    \label{fig:dp_diagram}
\end{figure}
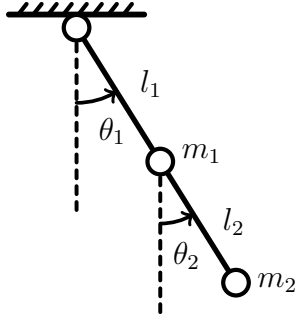

Finally, we consider a realistic setting in which the long-run covariance is unknown and must be estimated from data. For this study, we simulate two ensembles of double pendulum trajectories (Figure~\ref{fig:dp_diagram}) and track four observables: the angular positions $(\theta_1,\theta_2)$ and angular velocities $(\omega_1,\omega_2)$. The physical parameters are $L_1=L_2=2.0\,\mathrm{m}$ and $M_1=M_2=1.0\,\mathrm{kg}$. Each ensemble contains $1000$ trajectories, each of length $100{,}000$ time steps. Under Hamiltonian dynamics, energy is conserved, and invariant measures are therefore supported on fixed-energy surfaces. Consequently, comparisons of limiting distributions are meaningful only within a common energy level \cite{arovas:2013}.
We consider two ensembles at energies $E_1 = 70\,\mathrm{J}$ and $E_2 = 178\,\mathrm{J}$, with initial conditions sampled as described in Appendix~A. Empirically, the lower-energy ensemble explores phase space more uniformly and exhibits stronger convergence of all four observables, whereas the higher-energy ensemble shows non-negligible heterogeneity, particularly in the velocity coordinates.

We compute pairwise Wasserstein distances after a burn-in of $50{,}000$ steps. Since no closed-form covariance is available for the double pendulum observables, we estimate the long-run covariance using Definition~\ref{def:plug-in-cov-est} with bandwidth $L_n=250$ and flat-hat lag window $w_{n,k}=\lambda\!\left(k/(L_n+1)\right)$, where $\lambda(x)=1$ for $|x|\le 1/2$, $\lambda(x)=2(1-|x|)$ for $1/2<|x|\le 1$, and $\lambda(x)=0$ for $|x|>1$. 

The results where the trajectory distributions roughly converge to the same distribution are presented in Figure~\ref{fig:dp_pairwise_comparison}(a). Although not perfect, the theoretical limiting distribution and empirical distribution are on the same scale, have the same shape, comparable 95th quantile without a clear tendency to over- or underestimate, and the main difference occuring in the fatness of the tails. 
In the divergent case (Figure~\ref{fig:dp_pairwise_comparison}(b)), especially for $\omega_1$ and $\omega_2$, the empirical statistic is clearly multimodal, reflecting a mixture of trajectory pairs that converge to several invariant distributions. 
Table~\ref{tab:dp_rejection_comparison} shows strong power growth with increasing test level, with residual under-coverage attributable to the remaining mixture of convergent and divergent pair types; unlike the MA example, we do not pre-filter the data to remove pairs that converge to the same invariant distribution.

\begin{figure}[h!]
    \centering
    \begin{subfigure}{0.49\textwidth}
        \centering
        \includegraphics[width=\textwidth]{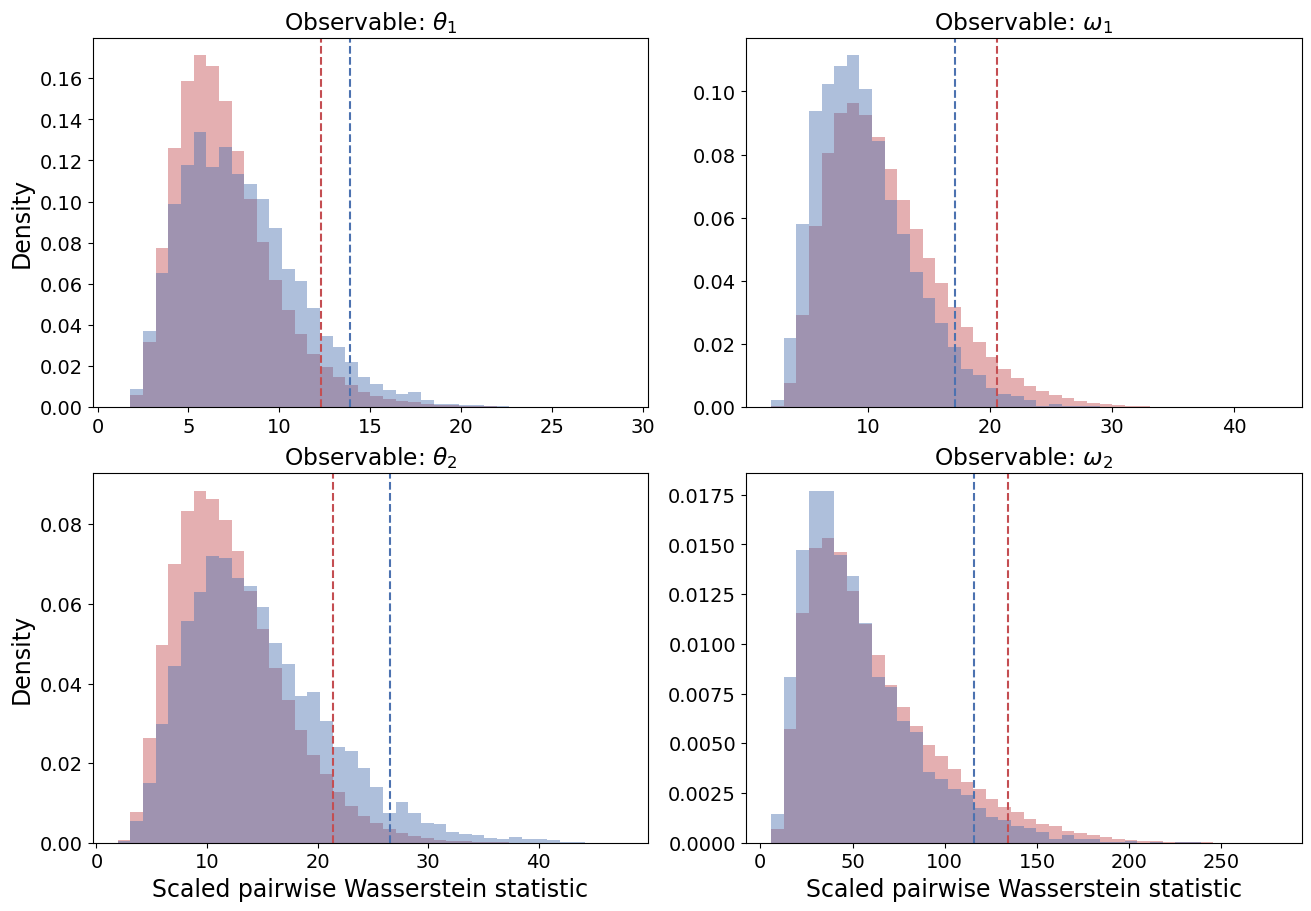}
        \caption{}
    \end{subfigure}
    \begin{subfigure}{0.49\textwidth}
        \centering
        \includegraphics[width=\textwidth]{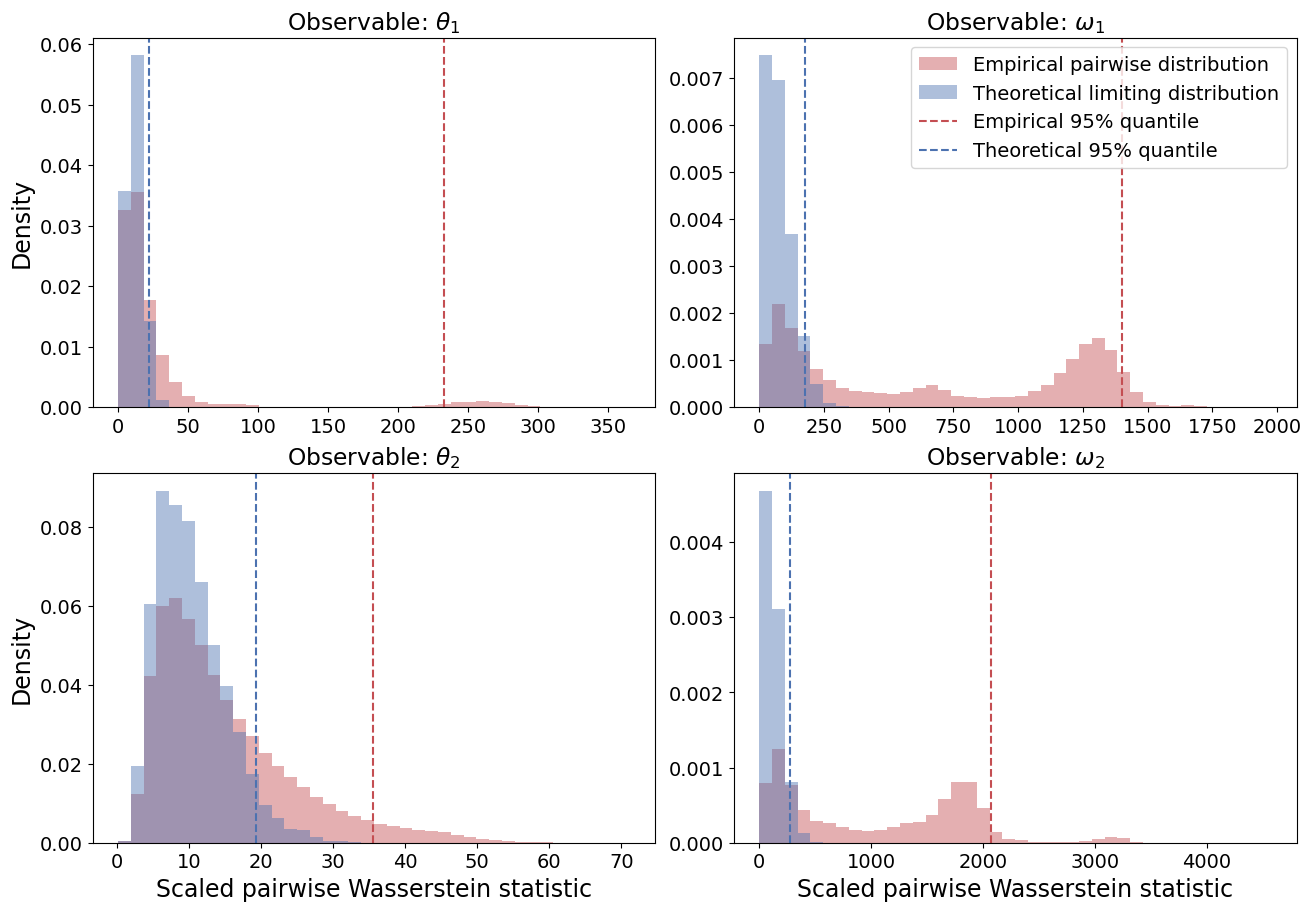}
        \caption{}
    \end{subfigure}
    \caption{Comparison of the distribution of the scaled empirical Wasserstein statistics $\sqrt{n}\,W_1\bigl(\hat\mu_n^{(i)},\hat\mu_n^{(j)}\bigr)$ and the theoretical limiting distribution for the four observables of a double pendulum. (a) Convergent setting: trajectories from the same invariant distribution; (b) Divergent setting: trajectories from different invariant distributions.}
    \label{fig:dp_pairwise_comparison}
\end{figure}

\begin{table}[ht]
\centering
\begin{minipage}[t]{0.47\textwidth}
\centering
\[
\text{\bf Convergent Case}
\]
\begin{tabular}{cccc}
\hline
\textbf{Obs.} & \boldmath$\alpha=0.01$ & \boldmath$\alpha=0.05$ & \boldmath$\alpha=0.10$ \\
\hline
$\theta_1$ & 0.5069\% & 2.4450\% & 5.0290\% \\
$\omega_1$ & 4.0110\% & 12.5642\% & 20.2120\% \\
$\theta_2$ & 0.0841\% & 1.0064\% & 2.9375\% \\
$\omega_2$ & 2.4865\% & 9.1481\% & 15.6985\% \\
\hline
\end{tabular}
\end{minipage}
\hfill
\begin{minipage}[t]{0.47\textwidth}
\centering
\[
\text{\bf Divergent Case}
\]
\begin{tabular}{cccc}
\hline
\textbf{Obs.} & \boldmath$\alpha=0.01$ & \boldmath$\alpha=0.05$ & \boldmath$\alpha=0.10$ \\
\hline
$\theta_1$ & 21.1313\% & 29.2358\% & 34.3369\% \\
$\omega_1$ & 65.2709\% & 70.5688\% & 73.9201\% \\
$\theta_2$ & 14.5251\% & 26.2192\% & 32.9009\% \\
$\omega_2$ & 66.4917\% & 72.1231\% & 76.1049\% \\
\hline
\end{tabular}
\end{minipage}
\caption{Empirical rejection rates (in \%) by observable and significance level, reported separately for convergent and divergent double pendulum ensembles.}
\label{tab:dp_rejection_comparison}
\end{table}

\section{Conclusion}
We introduced Wasserstein-based hypothesis tests for empirical-measure convergence in dependent sequences, covering both the one-sample setting with a known invariant measure and the practically important setting in which the invariant measure is unknown. Our main methodological contribution is the pairwise test based on $\sqrt{n}\,W_1\bigl(\hat\mu_n^{(i)},\hat\mu_n^{(j)}\bigr)$, together with asymptotic guarantees for coverage and power.

Across experiments, the asymptotic approximation is strongest when the long-run covariance structure is known explicitly, yielding good estimation and clear separation under alternatives. When the covariance must be estimated, performance remains useful, but finite-sample distortions are more visible. Our theoretical results in this regime are deliberately conservative: the finite-grid plug-in analysis establishes consistency of the covariance estimator and of the associated oracle grid-based Gaussian critical values, but it does not by itself prove asymptotic validity for the exact continuous $W_1$ statistic. Bridging that remaining gap requires a grid-refinement argument with $m=m_n\to\infty$ together with assumptions controlling the discretization error. This points to a clear practical direction: applications of the proposed tests should prioritize accurate long-run covariance estimation, while future work should couple that step with a continuous-limit refinement theory.

\section*{Acknowledgements}
    This research was conducted with the support of the Swiss National Science Foundation and the Bulgarian Ministry of Education and Science under Grant Agreement No. IZPYZ0\textunderscore228861 for the project ``Automatic maximally rigorous uncertainty quantification to enable design for certification'' (unCertify).

    The work is also part of the GATE project, funded under the Horizon 2020 WIDESPREAD-2018-2020 TEAMING Phase 2 program, grant agreement no. 857155 and the program ``Research, Innovation and Digitalization for Smart Transformation'' 2021-2027 (PRIDST) under grant agreement no. BG16RFPR002-1.014-0010-C01.

    We thank Prof.~Daniel Lathrop for his valuable comments on ergodicity and single-measure convergence in the double pendulum simulation.

\bibliographystyle{plain} 
\bibliography{bibliography} 

\newpage

\section*{Appendix A. Double Pendulum Simulation Initial Conditions Sampling Algorithm}
As discussed in \cite{arovas:2013}, it is appropriate to talk about ergodicity only for ensembles of Hamiltonian (i.e., measure-preserving) systems of the same total energy. Hence, given an initial energy $E$, we want to generate a set of initial conditions $(\theta_1, \omega_1, \theta_2, \omega_2)$. Below, we present the algorithm used for that. Note, that since we do not know the actual ensemble distribution of the phase space, we do not have any reason to put more emphasis on some regions of the space than others, hence, we use a uniform distribution for sampling. 

\noindent\textbf{Algorithm for Sampling Initial Conditions from a Constant Energy Surface}
\begin{enumerate}
    \item \textbf{Input:} The total energy $E$ and the system parameters, $m_1, m_2, l_1, l_2$, and $M = m_1 + m_2$.  
    
    \item \textbf{Random Angles:} Sample $\theta_1, \theta_2 \sim \operatorname{Unif} [-\pi, \pi]$ or $\theta_1, \theta_2 \sim \operatorname{Unif} [0, 2\pi]$

    \item \textbf{Potential Energy Check:} Calculate $$V(\theta_1, \theta_2) = -M g l_1 \cos(\theta_1) - m_2 g l_2 \cos(\theta_2)$$
    If $V(\theta_1, \theta_2) > E$, the configuration is physically impossible (requires negative kinetic energy), so discard the sample.
    
    \item \textbf{Determine Kinetic Energy Bounds:} To ensure a real solution for $\omega_2$, impose a maximum bound on $\omega_1$:
    $$(\omega_1)_{\mathrm{max}} = \sqrt{ \frac{2 (E-V)}{l_1^2 (m_1 + m_2 \sin^2(\theta_1 - \theta_2))} }$$
    
    \item \textbf{Sample $\omega_1$:} Sample $\omega_1 \sim \operatorname{Unif} [-(\omega_1)_{\mathrm{max}}, +(\omega_1)_{\mathrm{max}}]$
    
    \item \textbf{Calculate Coefficients:} Compute the coefficients for the quadratic equation using the sampled $\omega_1$:
    \begin{align*}
        A &= m_2 l_2^2 \\
        B &= 2 m_2 l_1 l_2 \omega_1 \cos(\theta_1 - \theta_2) \\
        C &= M l_1^2 \omega_1^2 - 2 (E-V)
    \end{align*}
    
    \item \textbf{Solve for $\omega_2$:} Calculate the two possible solutions for the second angular velocity using the quadratic formula:
    $$(\omega_2)_{\pm} = \frac{-B \pm \sqrt{B^2 - 4AC}}{2A}$$
    Randomly select either the $(+)$ or $(-)$ solution to ensure phase space is sampled evenly.
    
    \item \textbf{Return an Initial State:} The initial condition vector is $[\theta_1, \omega_1, \theta_2, \omega_2]$.
\end{enumerate}

\end{document}